\documentclass[10pt]{article}

\usepackage[left=3.5cm,right=3.5cm,top=3.5cm,bottom=3.5cm]{geometry}
\usepackage{bbm}
\usepackage{amsmath}
\usepackage{amsfonts}
\usepackage{amssymb}
\usepackage{amsthm}
\usepackage{graphicx}
\usepackage{enumerate}
\usepackage{dsfont}
\usepackage{multicol}
\usepackage{setspace}
\usepackage[round]{natbib}
\usepackage{authblk}
\usepackage{fancyhdr}
\newtheoremstyle{complement}
   {4ex}
   {4ex}
   {\scriptsize}
   {5ex}
   {\itshape}
   {.}
   { }
   {}

\theoremstyle{plain}
\newtheorem{theorem}{Theorem}[section]
\newtheorem{lemma}{Lemma}[section]

\newtheorem{property}[theorem]{Property}

\newcommand{\R}{\mathbb{R}}
\newcommand{\N}{\mathbb{N}}
\newcommand{\V}{\mathcal{V}}
\newcommand{\W}{\mathcal{W}}
\newcommand{\U}{\mathcal{U}}
\newcommand{\betapar}{\tilde{\beta}}

\begin{document}

\title{Conditional survival given covariates and marginal survival}

\author[1]{Roxane Duroux}
\affil[1]{Laboratoire de Statistique Th\'eorique et Appliqu\'ee,
 Universit\'e Pierre et Marie Curie, Paris, France\\}

\author[2]{C\'ecile Chauvel}
\affil[2]{Laboratoire Jean Kuntzmann, Universit\'e Joseph Fourier, Grenoble, France\\}

\author[1]{John O'Quigley}

\date{}

\maketitle

\begin{abstract}
Assuming some regression model, it is common to study the conditional distribution of survival given covariates. Here, we consider the impact of further conditioning, specifically conditioning on a marginal survival function, known or estimated. We investigate to what purposes any such information can be used in a proportional or non-proportional hazards regression analysis of time on the covariates. It does not lead to any improvement in efficiency when the form of the assumed proportional hazards model is valid. However, when the proportional hazards model is not valid, the usual partial likelihood estimator is not consistent and depends heavily on the unknown censoring mechanism. In this case we show that the conditional estimate that we propose is consistent for a parameter that has a strong interpretation independent of censoring. Simulations and examples are provided.
\end{abstract}


\section{Introduction}

Prior to the development of the Cox proportional hazards model \citep{cox1972}, parametric models for carrying out regression on censored survival data enjoyed much success. Even the simple exponential model proved itself to be reliable and valuable in a broad range of situations including those in which the exponential assumption itself was no more than a rough approximation. Nonetheless, these parametric models were quickly almost entirely eclipsed by the arrival of the Cox model. The advantages of Cox regression can be seen on at least two levels. First the user is entirely freed of the need to consider plausible forms for the baseline hazard, \textit{i.e.} the conditional distribution of time, $T$, given a scalar or vector covariate $Z$. In addition, the Cox model immediately extends to much more involved situations such as time dependent covariates, multi-state processes and recurrent events. Furthermore, inference is invariant to monotonic increasing transformations on $T$ and efficiency remains high. In other words relatively little information has been lost by using this model that makes inference on the relative risk parameters easier. 

Suppose however that we would like to include in our analysis the marginal survival function of $T$, that we denote $S(t)$. There can be many situations in which we know something about the marginal distribution of $T$, prior to undertaking any regression analysis of $T$ on $Z$. We may have an accurate estimate of $S$, for example using registry data. We may wish to calibrate to some other study in which, by hypothesis, the mechanism governing the generation of the random variable $T$ is the same. Finally, prior to the regression analysis, we may decide to fit a marginal model to the distribution of $T$. This is the case of main interest here. Any unknown parameters for this marginal distribution are then replaced by sample based estimates. We then treat these estimated parameters as though they were fixed and known.

The question that we address in this paper is the extent to which any such undertaking is of value. In the light of previous works concerning the efficiency of the partial likelihood estimates, we do not anticipate any significant gains there \citep{struthers-kalbfleisch1986,lin1991,xu-oquigley2000}. One valuable result is that we can misspecify the model for $S$, failing to correctly model the distribution of $T$, and yet maintain consistent estimates for the relative risk parameters. In such a case there would be efficiency losses but these appear to be small. We might then argue that we are not risking much by the suggested approach. On the other hand, we can make very useful gains when the model for relative risk itself is misspecified. Suppose, for instance, that instead of a constant log relative risk, $\beta$, as supposed by the proportional hazards model, the observations are generated by a more general set-up in which $\beta$ changes over time. To make this more precise, we use the notation $\beta (t)$ \citep{murphy-sen1991}. The partial likelihood estimator will converge to a quantity depending in a very involved way on the unknown censoring mechanism, even when independent of both the failure mechanism and the covariate \citep{struthers-kalbfleisch1986}. This dependence is very strong and has been noticed by a number of authors. In situations of non-proportional hazards, in particular in cases of a smooth change through time, we would like to know if we are able to estimate $E[ \beta (T)]$ using the marginal survival. In fact, it turns out that this can be easily obtained by conditioning on the observed marginal survival estimate of $F(t)$.

The Cox proportional hazards model \citep{cox1972} allows for us to make inference on the regression coefficients of a relative risk model while keeping unspecified some baseline hazard rate, $\lambda_0$. A non-proportional hazards model, of which
Cox's model would be a special case, can be written,
\begin{eqnarray}\label{1.1}
\lambda(t|Z)=\lambda_0(t)\exp\{\beta(t) Z\},
\end{eqnarray}
and, in either case, the baseline hazard $\lambda_0$ is to be interpreted as the hazard $\lambda (t|Z=0)$. Model (\ref{1.1}) has been looked at by \citet{moreau-oquigley-mesbah1985,oquigley-pessione1989,oquigley-pessione1991,liang-self-liu1990,zucker-karr1990,murphy-sen1991,gray1992,hastie-tibshirani1993,verweij-houwelingen1995,lausen-schumacher1996,marzec-marzec1997}, and references therein.
The main emphasis of these papers was to estimate the regression effect $\beta(t)$ as a function of $t$. The goal of estimating $\beta(t)$ represents a considerable challenge since, in general situations, $\beta(t)$ is of infinite dimension. A less ambition although more readily achievable goal is to estimate the average effect $E[\beta(T)]$. It turns out that conditioning on marginal survival leads to an immediate solution to this problem. Estimation of an average effect can be used in a preliminary analysis of a data set with time varying regression effects.
Note that when using standard software such as SAS or R, the user may guess that the single estimate $\hat{\beta}$ in cases where $\beta(t)$ varies with $t$, corresponds to an average with respect to the variable $T$. This is in fact not true, because of the dependence on the censoring mechanism \citep{xu-oquigley2000}.

In Section \ref{Inference} we derive an estimate $\betapar$ based on the knowledge of the marginal survival function. When the proportional hazards assumption holds, this estimate is consistent for the ``true" regression parameter. The estimate is easy to compute in practice. We study the large sample properties of $\betapar$ and give the interpretation of $\beta^*$, to which $\betapar$ converges in probability, as a population average effect in Section \ref{LS_prop}. Simulations are provided in Section \ref{Simu}. The relative efficiency of $\betapar$ to the partial likelihood estimate under the proportional hazards model is studied in Section \ref{Rel-eff}. Finally Section \ref{Appli} gives an example of how $\betapar$ can be used in practice.

\section{Inference and conditional inference given $F(t)$}\label{Inference}

We denote $(T_i, C_i, Z_i(.))$ for $i \in \{1, \cdots , n\}$ a sequence of i.i.d.random variables with the same distribution as $(T,C,Z(.))$, where $T$ is the random variable of distribution function $F$ representing the failure time, $Z(.) \in \R^d$ is the covariate vector and $C$ is the censoring time, independent of $T$ given $Z(.)$. Let us assume that there exists $\tau >0$ such that $[0, \tau]$ is the support of $T$ and $C$. We also assume that $T$ and $Z$ follow model (\ref{1.1}) with parameter $\beta_0$. Let us define for all $i \in \{1, \cdots, n\}$, $X_i= \min(T_i,C_i)$ and $\Delta_i= I(T_i \leq C_i)$ so that $X_i$ is the observed time for the patient $i$ and $\Delta_i$ is the assigned status to this patient: ``died" ($=1$) or ``censored" ($=0$). The covariate of patient $i$ is observed until time $X_i$ but we extend its definition on $[0, \tau]$ by denoting $Z_i(s)=Z_i(X_i)$ for all $ s \in [X_i, \tau]$. Then we denote $\mathcal{D}$ the set of functions that are right continuous with left-hand limits from $[0, \tau]$ to $\R^d$, and we have $\beta_0 \in \mathcal{D}$. We define for all $i \in \{1, \cdots ,n\}$ and all $t \in [0, \tau]$, $Y_i(t)= I(X_i \geq t)$. The random process $Y_i(t)$ indicates whether the patient $i$ is still at risk at time $t$ ($=1$), or not ($=0$). Finally, we denote by $F_n$ the empirical distribution of $T$, \textit{i.e.} $F_n(t)=\frac{1}{n} \sum_{i=1}^n I(T_i \leq t)$, for all $t \in [0, \tau]$. Before introducing our estimator, we recall some estimation results in proportional and non-proportional hazards models.


For all $r \in \{0,1,2\}$ and $\beta \in \mathcal{D}$, let us define
\[
S^{(r)}(\beta,t)= \frac{1}{n} \sum_{j=1}^n Y_j(t) \exp \left\{ \beta'Z_j(t) \right\} Z_j(t)^{\otimes r},
\]
where for a column vector $v$, $v^{\otimes 2}$ is the matrix $vv'$, $v^{\otimes 1}$ the vector $v$ and $v^{\otimes 0}$ the scalar 1. We write the log partial likelihood \citep{cox1972} 
\[
l(\beta)= \sum_{i=1}^n \Delta_i \left[ \beta'Z_i(X_i) - \log \left\{ S^{(0)}(\beta,X_i) \right\} \right],
\]
and define
\[
E(\beta,t)= \frac{S^{(1)}(\beta,t)}{S^{(0)}(\beta,t)}, \ \ \ V(\beta,t)=\frac{S^{(2)}(\beta,t)}{S^{(0)}(\beta,t)}- E(\beta,t)^{\otimes 2}.
\]
The score $U$ is a function of the parameter $\beta$ defined by
\begin{eqnarray}\label{eq_estim_PL}
U(\beta)= \sum_{i=1}^n \Delta_i \left\{Z_i(X_i) - E(\beta,X_i) \right\} = \int \left\{ \mathcal{Z}(t) - E(\beta,t) \right\} d \bar{N}(t),
\end{eqnarray}
where $\bar{N}(t) = \sum_{i=1}^n \mathds{1}\{X_i \leq t, \ \Delta_i=1 \}$ and $\mathcal{Z}(.)$ is a left-continuous step function with discontinuities at the points $X_i$ where it takes the value $Z_i(X_i)$. Note that, in absence of censoring,
\begin{eqnarray}\label{no_cens}
U(\beta)= \int \left\{ \mathcal{Z}(t) - E(\beta,t) \right\} d \bar{N}(t)= \int \left\{ \mathcal{Z}(t) - E(\beta,t) \right\} dF_n(t).
\end{eqnarray}
This expression allows us to investigate the case when $F(t)$ is known. To see this, note that $F_n(t)$ is consistent for $F(t)$. To maximize the partial likelihood of model (\ref{1.1}) under the assumption that $\beta_0$ is constant, we solve the estimating equation $U(\beta)=0$. Thus we get a real consistent estimator $\hat{\beta}_{PL}$ of the function $\beta_0$. If $\beta_0$ is indeed a constant function, \textit{i.e.} $\beta_0(t)=\beta_0$ for all $t \in [0, \tau]$, \citet{andersen-gill1982} showed that $\sqrt{n}(\hat{\beta}_{PL} - \beta_0)$ is asymptotically normal with mean zero and variance consistently estimated by $\mathcal{I}^{-1}(\hat{\beta}_{PL})$ where $\mathcal{I}(\hat{\beta}_{PL})= n^{-1} \sum_{i=1}^n \Delta_i V(\hat{\beta}_{PL},X_i)$.
However, if the assumption of a constant $\beta_0$ fails, it is helpful to obtain some summary measure of the whole function $\beta_0(t)$, such as the mean effect, and study how to estimate this quantity. We notice that in the score function (\ref{eq_estim_PL}), all the terms in the sum have the same weight. We may choose to assign different weights. This is what we do for instance when using a weighted log-rank test \citep{gehan1965}. 

The idea is to replace $F_n$, which is only available in the absence of censoring, in (\ref{no_cens}) by $\tilde{F}$ which provides a consistent estimator of $F$ under some conditions and in the presence of censoring. This leads us to define the following weighted score function
\begin{eqnarray}\label{eq_estimW}
U_{W}(\beta)= \sum_{i=1}^n \Delta_i W(X_i) \left\{ Z_i(X_i)-E(\beta,X_i) \right\},
\end{eqnarray}
where $W(.)$ is a real $\left(\mathcal{F}_t \right)_{t>0}$ - predictable stochastic process with, for $t \in [0, \tau]$,
\[\mathcal{F}_t = \sigma ((X_i, \Delta_i, Z_i(s)); \ i : X_i \leq s, \ 0 \leq s \leq t).\]
Loosely speaking, $\left(\mathcal{F}_t \right)_{t>0}$ is the filtration that contains all the observed information before time $t$. With a specific choice of $W$, we can find
\begin{eqnarray}
U_{W}(\beta)= \int \left\{ \mathcal{Z}(t) - E(\beta,t) \right\} d \tilde{F}(t).
\end{eqnarray}
Let denote $\hat{\beta}_W$ the solution of the equation $U_W(\beta)=0$. We see that if we choose $W$ constant equal to $1$, then we find the usual score function (\ref{eq_estim_PL}) arising from the partial likelihood.

\section{Statistical properties of the estimator}\label{LS_prop}

There are two ways of approaching large sample inference: conditional and marginal inference. For the first of these, we suppress all uncertainty in the estimate of $S(t)$. Such an approach may be appropriate when we wish for the marginal survival to reflect some given reference population. The second approach, marginal inference, takes on board uncertainties in any prior estimation of $S(t)$. These uncertainties can be of the form of errors in estimates or Bayesian when specified via some prior distribution.

\subsection{Conditional inference}

Let us denote $\hat{S}$ the left-continuous version of the estimator of \citet{kaplan-meier1958} of the survival function $S=1-F$ of the failure time marginal distribution $T$. The function $\hat{S}$ is a consistent estimator for $S$ \citep{xu-oquigley2000}. We now choose the weighted function $W_{KM}$ in equation (\ref{eq_estimW}) such that, for all $t \in [0,\tau]$,
\[
W_{KM}(t)=\frac{\hat{S}(t)}{\sum_{i=1}^n Y_i(t)}=\frac{\hat{S}(t)}{nS^{(0)}(0,t)}.
\]
A calculation gives the formula $\hat{S}= \sum_{i:X_i \leq t} \delta_i W_{KM}(X_i)$, such that the weights $\delta_i W_{KM}(X_i)$ are the increments of $\hat{S}$. We see that $W_{KM}$ converges uniformly in $t$ to $w_{KM}$ where $w_{KM}(t) = S(t) / s^{(0)}(0,t)$. Let us denote $\hat{\beta}_{KM}$ the estimator defined by the estimating equation $U_{W_{KM}}(\beta)=0$.

To study the asymptotic behaviour of $\hat{\beta}_W$, we now consider the following assumptions, the first three of which are due to \citet{andersen-gill1982} ; \textit{
\begin{enumerate}[(A)]
\item (Finite interval) $\int_0^\tau \lambda_0(t) dt < \infty$.
\item (Asymptotic stability) There exists a neighborhood $\mathcal{B}$ of $\beta_0$ such that $0$ and $\beta_0$ are in the interior of $\mathcal{B}$, and there exist functions $s^{(r)}(\beta,t)$ defined on $\mathcal{B} \times [0,\tau]$ for $r \in \{0,1,2\}$ such that
\begin{eqnarray*}
\sup_{\beta \in \mathcal{B}, t \in [0,\tau]} \Vert S^{(r)}(\beta ,t)-s^{(r)}(\beta ,t) \Vert \overset{P}{\longrightarrow} 0.
\end{eqnarray*}
\item (Asymptotic regularity conditions) For all $r \in \{0,1,2\}$, the functions $s^{(r)}(\beta,t)$ are uniformly continuous in $t \in [0,\tau]$, continuous in $\beta \in \mathcal{B}$ and bounded on $\mathcal{B} \times [0,\tau]$; $s^{(0)}(\beta,t)$ is bounded and bounded away by zero.
\item (Asymptotic stability of $W$) There exists a nonnegative bounded function $w$ defined on $[0,\tau]$ such that
\begin{eqnarray*}
\sup_{t \in [0,\tau]} \vert nW(t)-w(t) \vert \overset{P}{\longrightarrow} 0.
\end{eqnarray*}
\item (Homoscedasticity) If we let
\[
\frac{s^{(1)}(\beta_0,t)}{s^{(0)}(\beta_0,t)}= E[Z(t) \vert T=t], \ \text{and} \ v(\beta_0,t) = \left. \frac{\partial}{\partial \beta} \frac{s^{(1)}(\beta,t)}{s^{(0)}(\beta,t)} \right\vert_{\beta = \beta_0} = \text{Var}[Z(t) \vert T=t],
\]
we assume that $v(\beta_0,t)$ is constant in time.
\end{enumerate}
}

Let us take some other notations: let us define $\beta_w^*$ as the unique solution of the equation 
\begin{eqnarray}\label{eq_asymp}
h_w(\beta)= \int_0^{\tau} w(t) \left\{ \frac{s^{(1)}(\beta_0,t)}{s^{(0)}(\beta_0,t)} - \frac{s^{(1)}(\beta,t)}{s^{(0)}(\beta,t)} \right\} s^{(0)}(\beta_0,t)dt=0,
\end{eqnarray}
and $A_w(\beta)$ as 
\[
A_w(\beta)= \int_0^{\tau} w(t) \left\{ \frac{s^{(2)}(\beta,t)}{s^{(0)}(\beta,t)}- \left( \frac{s^{(1)}(\beta_0,t)}{s^{(0)}(\beta_0,t)} \right)^{\otimes 2} \right\} s^{(0)}(\beta_0,t)dt.
\]

\citet{lin1991} showed that
\begin{property}\label{Lin91}
Under model (\ref{1.1}) with parameter $\beta_0$ and the assumptions $(A)$, $(B)$, $(C)$ and $(D)$, $\hat{\beta}_W \overset{P}{\longrightarrow} \beta^*_w$ if $A_w(\beta^*_w)$ is positive definite.
\end{property}
In particular, choosing $w(t)=1$, we get that $\hat{\beta}_{PL}$ converges to the solution of the equation
\begin{eqnarray}\label{eq_asympPL}
\int_0^{\tau} \left\{ \frac{s^{(1)}(\beta_0,t)}{s^{(0)}(\beta_0,t)} - \frac{s^{(1)}(\beta,t)}{s^{(0)}(\beta,t)} \right\} s^{(0)}(\beta_0,t) \lambda_0(t)dt=0.
\end{eqnarray}
In general, this solution depends on the censoring through the term $s^{(0)}(\beta_0,t)$ \citep{struthers-kalbfleisch1986}. Therefore, the results of the estimation stemming from equation (\ref{eq_estimW}) have to be read with care and their interpretation is not straightforward in general situations.

Using Property \ref{Lin91}, we have the following result, due to \citep{xu-oquigley2000},
\begin{property}\label{Xu96}
Under model (\ref{1.1}) of parameter $\beta_0$ and under the assumptions $(A)$, $(B)$ and $(C)$, if $A_{w_{KM}}(\beta^*)$ is positive definite, we have $\hat{\beta}_{KM} \underset{n \to \infty}{\overset{P}{\longrightarrow}} \beta_{w_{KM}}^*$.
\end{property}
For sake of simplicity, let us denote $\beta^*=\beta_{w_{KM}}^*$. The parameter $\beta^*$ is the unique solution of the equation $h_{w_{KM}}(\gamma)=0$,
\textit{i.e.} the equation
\begin{eqnarray}\label{eq_asympKM}
\int_0^{\tau} \left\{ \frac{s^{(1)}(\beta_0,t)}{s^{(0)}(\beta_0,t)} - \frac{s^{(1)}(\beta,t)}{s^{(0)}(\beta,t)} \right\} dF(t)=0.
\end{eqnarray}
None of the quantities in Equation (\ref{eq_asympKM}) involves the censoring mechanism, and so, its solution $\beta^*$ is independent of the censoring. 

Now, let us show that $\beta^*$ can be interpreted as the average effect of the regression function $\beta_0$.
We apply a Taylor series approximation to (\ref{eq_asympKM}) and we get
\[
\int_0^\tau \text{Var}[Z(t) \vert T=t] \{\beta^*-\beta_0(t) \}dF(t) \approx 0,
\]
which can be written directly in terms of $\beta^*$ as
\[
\beta^* = \frac{\int_0^\tau \text{Var}[Z(t) \vert T=t] \beta_0(t)dF(t)}{\int_0^\tau \text{Var}[Z(t) \vert T=t] dF(t)} = \frac{\int_0^\tau v(t) \beta_0(t)dF(t)}{\int_0^\tau v(t) dF(t)}.
\]
$\beta^*$ is a weighted average effect of the regression coefficient with weights proportional to $v(.)$. Now, let us make a first approximation that $v(.)$ is nearly constant. Then we obtain
\begin{property} Under the homoscedasticity assumption $(E)$,
\begin{eqnarray}\label{esp_beta}
\beta^* \approx \frac{\int_0^\tau \beta_0(t) dF(t)}{F(\tau)} = \int_0^\tau \beta_0(t) dF(t) = E[\beta_0(T)].
\end{eqnarray}
\end{property}
We recall that $F(\tau)=1$ because $[0, \tau]$ is the support of $T$. In the case of the \citet{harrington-fleming1982} models, the equation (\ref{esp_beta}) holds exactly. To see this, consider a relation between $\beta^*$ and a vector $\alpha$ measuring group differences in $k$-sample transformation models when the random error belongs to the $G^\rho$ family of \citet{harrington-fleming1982}. Note that, such a transformation model can be written $h(T)=\alpha'Z + \varepsilon$ where $h$ is an increasing function, $Z$ a vector with values in $\{0,1\}$ and $\varepsilon$ belongs to the $G^\rho$ family. The survival function of $Z$ is defined by
\[
\left\{
\begin{matrix}
H_0(t)=\exp(-e^t) & (\rho=0)\\
H_\rho(t)=(1+\rho e^t)^{1/\rho} & (\rho>0).
\end{matrix}
\right.
\]
The non-proportional hazards model and the above transformation model are equivalent with $h(t)=\log \Lambda_0(t)$ where $\Lambda_0$ is the cumulative baseline hazard and $\alpha=\beta_0(0)$. \citet{xu-harrington2001} showed that, in the two-sample case, with $\rho=1$, $h=\log$ and equal probabilities of group membership, \textit{i.e.} $P(Z=1)=1-P(Z=0)=1/2$, we have
\[
E[\beta_0(T)]=-\frac{\alpha}{2}=\beta^*,
\]
where the first equality is straightforward and the second one uses the fact that
\[
E[Z \vert T=t]=\frac{s^{(1)}(\beta_0,t)}{s^{(0)}(\beta_0,t)}.
\]
In this way we can consider $\beta^*$, as a first approximation, to an average effect of the regression coefficient.

The estimator $\hat{\beta}_{KM}$ presents better performances than those of the estimators $\hat{\beta}_{PL}$ and $\hat{\beta}_W$ in the case of non proportional hazards. In addition, since this estimator does not depend on censoring, its use in pratical situations is justified and its interpretation is clear. However the jumps of the Kaplan-Meier estimator are large for large survival times when there is significant censoring. This is implied by  a decrease in the size of the set of individuals at risk. In this way, we put more weight on delayed observations in the estimating equation. Often, the last increments are particularly noisy \citep{stute-wang1993}. 

One way to get around these problems is to use smooth weights in the estimating equation (\ref{eq_estimW}) by appealing to a parametric model for $T$. Let us define the parametric model $\mathcal{M}=\{P_\theta , \theta \in \Theta \}$ where $\Theta \subset \R^p$ and $P_\theta$ is a continuous distribution function for all $\theta$. Let us assume that the distribution function of $T$ is $P_{\theta_0}$, with a fixed $\theta_0 \in \Theta$. Then we denote $S=S_{\theta_0}=1-P_{\theta_0}$ its survival function. We denote $\hat{\theta}_n$ a consistent estimator of $\theta_0$, \textit{i.e.} we have the convergence
\begin{equation}\label{conv_estimateur}
\hat{\theta}_n \overset{P}{\longrightarrow} \theta_0.
\end{equation} 
Let $T_n$ denote a random variable with distribution function $P_{\hat{\theta}_n}$ and $S_n=1-P_{\hat{\theta}_n}$ its survival function. Let us define the weight based on the parametric model
\[
W_{p}(t)=\frac{S_n(t)}{nS^{(0)}(0,t)}.
\]
We obtain the following result.

\begin{theorem}\label{NPH_bien_spe}
Under model (\ref{1.1}) with parameter $\beta_0$ and under the assumptions $(A)$, $(B)$ and $(C)$, if $A_w(\beta^*)$ is positive definite, then $\betapar \overset{P}{\longrightarrow} \beta^*$, where $\betapar$ is the solution of the equation $U_{W_{p}}(\beta)=0$.
\end{theorem}

\begin{proof}
It is sufficient to show that $W_{p}$ satisfies the assumption $(D)$ to obtain the result using Property \ref{Lin91}. We recall the following lemma which we make use of.

\begin{lemma}\label{Dini_prob}(Probabilist Dini theorem)\\
Let $(X_n(t))_{n \in \N}$ denotes a sequence of almost surely real continuous non-decreasing processes defined on $I=[a,b] \subset \R$.\\
If, for all $t \in I$, $X_n(t) \overset{P}{\rightarrow} x(t)$, when $n$ tends to infinity, where $x$ is a continuous function on $I$, then 
\[
\sup_{t \in I} \Vert X_n(t) - x(t) \Vert \overset{P}{\rightarrow} 0.
\]
\end{lemma}

Since $w$ is continuous, it is uniformly continuous on $[0,\tau]$.
Then we show that 
\[
\sup_{t \in [0,\tau]} \vert nW_{p}(t)-w(t) \vert \overset{P}{\longrightarrow} 0,
\]
thanks to
\begin{equation}\label{conv_S_n}
\sup_{t \in [0,\tau]} \vert S_n(t)-S_{\theta_0}(t) \vert \overset{P}{\longrightarrow} 0
\end{equation}
\begin{equation}\label{conv_s^0}
\sup_{t \in [0,\tau]} \vert S^{(0)}(0 ,t)-s^{(0)}(0 ,t) \vert \overset{P}{\longrightarrow} 0,
\end{equation}
and the boundedness conditions of assumption $(C)$, where (\ref{conv_S_n}) comes from Lemma \ref{Dini_prob} and (\ref{conv_s^0}) from assumption $(B)$.
Finally, $W_{p}$ satisfies $(D)$ and we can use Property \ref{Lin91}, which ends the proof of Theorem \ref{NPH_bien_spe}.
\end{proof}

We have estimated $\beta^*$ with different weights than $W_{KM}$ in equation (\ref{eq_estimW}), using our knowledge of the marginal survival. This parameter is of interest because of its lack of dependence on the censoring mechanism and its closeness to $E[\beta_0(T)]$.

One can see that $\hat{\theta}_n$ satisfies Equation (\ref{conv_estimateur}) when the marginal distribution of $T$ is well-specified. In this case, the parameter $\theta_0$ of the marginal distribution of $T$ can be estimated by maximizing the likelihood or by the method of moments before focusing our interest on the consistent estimation of $\beta^*$ consequent upon Theorem \ref{NPH_bien_spe}. As mentioned previously, it may be of interest to consider a smooth estimator of the survival function $S$ of $T$, which is not too noisy for late observed times. Finally, if we have information from previous studies, it is reasonable to use this information for the estimation. Theorem \ref{NPH_bien_spe} ensures that we can adjust the weights in the score function (\ref{eq_estimW}) with $W_{p}$, and then the estimator converges to $\beta^*$, which we take to be an estimate of average effect.

\subsection{Marginal inference}

Let $g$ be the density of the prior distribution on $\theta$. This distribution can summarize errors of estimation on $\theta$ or be a Bayesian prior on this variable. We assume that $\Theta$ is a compact set of $\R^p$. We denote $g_n$ the density of $\theta_n$ which is the $n$-th bayesian estimation of $\theta$, using the likelihood. For each fixed $\theta$, we can consider that $T$ follows the distribution $F(t;\theta)$ and we define
\begin{eqnarray}\label{score_theta_n}
U_n(\theta, \beta) &=& \int \left\{ \mathcal{Z}(t) - E(t,\beta) \right\} dF_n(t,\theta)\\
&=& \sum_{i=1}^n \Delta_i \frac{S(X_i;\theta)}{nS^{(0)}(0,X_i)} \left\{ Z_i(X_i) - E(\beta,X_i) \right\},
\end{eqnarray}

where $S(t;\theta)=1-F(t;\theta)$. We assume that $S(t;.)$ is continuously differentiable for all $t \in [0,\tau]$ and bounded away by zero. For each $\theta$, we can apply Property \ref{Lin91} and define the sequence $(\beta_{n,\theta})_n$ such that for all $n \in \N$, $U_n(\theta,\beta_{n,\theta})=0$ and $\beta_{n,\theta} \overset{P}{\rightarrow} \beta_\theta$, where $\beta_\theta$ is the solution of the equation
\begin{equation}\label{eq-lim-theta}
h_\theta(\beta) = \int_0^\tau w_\theta(t) \left\{ \frac{s^{(1)}(\beta_0,t)}{s^{(0)}(\beta_0,t)} - \frac{s^{(1)}(\beta,t)}{s^{(0)}(\beta,t)} \right\} s^{(0)}(\beta_0,t)dt=0,
\end{equation}
with $w_\theta(t)=S(t;\theta)/s^{(0)}(0,t)$. We use then Lemma \ref{TFI_prob} written in Appendix.

It allows us to consider the $\mathcal{C}^1$-process $\left(\beta_n(\theta) \right)_n$ defined on $\Theta$ such that, for all $\theta \in \Theta$, $U_n(\theta,\beta_n(\theta))=0$ and $\beta_n(\theta) \overset{P}{\rightarrow} \beta(\theta)$. Besides we have, using \eqref{diff-fct-imp} in Lemma \ref{TFI_prob},
\begin{eqnarray}\label{deriv-beta}
\beta_n'(\theta)=-\left(\frac{\partial U_n}{\partial \beta} (\theta, \beta_n(\theta)) \right)^{-1} \frac{\partial U_n}{\partial \theta} (\theta, \beta_n(\theta)).
\end{eqnarray}

We can study the two partial derivatives in (\ref{deriv-beta}). On one hand, we have
\begin{eqnarray*}
\frac{\partial U_n}{\partial \theta}(\theta, \beta_n(\theta)) = \sum_{i=1}^n  \frac{\Delta_i}{nS^{(0)}(0,X_i)} \frac{\partial S}{\partial \theta}(X_i;\theta) \left\{ Z_i(X_i) - E(\beta_n(\theta),X_i) \right\}.
\end{eqnarray*}

Assumptions $(B)$, $(C)$ and the $\mathcal{C}^1$-differentiability of $S(t;.)$ imply that, at least for $n$ sufficiently high, $\partial U_n/\partial \theta$ is bounded by a constant independent of $n$. On the other hand,

\begin{eqnarray*}
\frac{\partial U_n}{\partial \beta}(\theta, \beta_n(\theta))
= - \sum_{i=1}^n \Delta_i \frac{S(X_i; \theta)}{n S^{(0)}(0,X_i)} V(\beta_n(\theta),X_i).
\end{eqnarray*}

Assumptions $(B)$, $(C)$, $(E)$, the fact that $S(t,.)$ is bounded away by zero and the strong law of large number applied on $\left(\Delta_i \right)_{i \in \{1, \cdots, n \}}$ imply that, at least for $n$ sufficiently high, $\partial U_n/\partial \beta$ is bounded and majored by a negative constant independent of $n$. So, with (\ref{deriv-beta}), $\beta_n'(.)$ is almost surely bounded by a constant $K>0$ independent of $n$. We have then that $\beta_n$ is almost surely uniformly Lipschitz continuous and so 
\begin{eqnarray}\label{conv-unif-beta}
\Vert \beta_n - \beta \Vert_{\infty} \overset{P}{\longrightarrow} 0.
\end{eqnarray}

By definition of $\theta_n$, we have
\begin{eqnarray}\label{conv-theta-bayes}
\theta_n \overset{P}{\longrightarrow} \theta_0.
\end{eqnarray}
If we combine (\ref{conv-theta-bayes}) and (\ref{conv-unif-beta}) with a triangle inequality, we obtain
\begin{eqnarray}\label{conv-bayes}
\beta_n(\theta_n) \overset{P}{\longrightarrow} \beta(\theta_0)=\beta^*.
\end{eqnarray}
This convergence enables us to consider a bayesian inference for the study of an average effect of $\beta_0$.

In practice, at step $n$, we begin with the bayesian estimation of $\theta_n$. Then we estimate $\beta$ using the equation $U(\theta_n,\beta)=0$ where
\begin{eqnarray*}
U(\theta_n, \beta) = \sum_{i=1}^n \Delta_i \frac{S(X_i;\theta_n)}{nS^{(0)}(0,X_i)} \{ Z_i(X_i) - E(\beta,X_i) \}.
\end{eqnarray*}
This equation can be solved with the Newton-Raphson method. Equation (\ref{conv-bayes}) implies that this process will converge to $\beta^* \approx E[\beta(T)]$.

\section{Simulations}\label{Simu}

Simulations were carried out to study $\betapar$ and compare it with the partial likelihood estimator $\hat{\beta}_{PL}$ and the Kaplan-Meier estimator $\hat{\beta}_{KM}$.
First, we study some cases under proportional hazards models, \textit{i.e.} a model with regression parameter $\beta_0(t)=\beta_0$ for all $t \in [0,\tau]$ ; then some cases under non-proportional hazards models. More specifically, we consider changepoint models with a piecewise constant regression coefficient : $\beta_0(t)= \beta_1 I(t < t_0) + \beta_2 I(t \geq t_0)$, for all $t \in [0, \tau]$. The sample size is $n=1500$. Some of these simulations are presented in this paper and many others can be provided by the authors. The failure time $T$ follows an exponential distribution of parameter 2. The covariate $Z$ follows a uniform distribution on $[0,1]$. The censoring time $C$ follows a uniform distribution on $[0,t_c]$ (Table \ref{tab1} and Table \ref{tab2}) or an exponential distribution of parameter $t_c$ (Table \ref{tab3}), where $t_c$ is set to fix the percentage of censoring. We carried out 500 simulations for each case and computed the empirical means and standard errors of the estimators.

\begin{table}[!h]
\centering
\caption{Comparison of $\hat{\beta}_{PL}$, $\hat{\beta}_{KM}$ and $\betapar$ under proportional hazards model. $C \sim \mathcal{U}[0,t_c]$. Standard errors in parenthesis.}
\label{tab1}
\begin{tabular}{|c|c|c|c|c|c|}
\hline
$\beta_0$ & $\%$ of censoring & $\hat{\beta}_{PL}$ & $\hat{\beta}_{KM}$ & $\betapar$ & $E[\beta_0(T)]$\\
\hline
1 & 0\% & 1.000 (0.117) & 1.000 (0.117) & 1.000 (0.117) & 1 \\
 & 50\% & 0.996 (0.115) & 1.001 (0.182) & 1.000 (0.186) & 1\\
\hline
0.5 & 0\% & 0.498 (0.103) & 0.498 (0.103) & 0.498 (0.103) & 0.5\\
 & 50\% & 0.503 (0.100) & 0.505 (0.184) & 0.506 (0.190) & 0.5\\
\hline 
\end{tabular}
\end{table}

The results in Table \ref{tab1} indicate that the three estimators are performing well under proportional hazards models. Moreover $\hat{\beta}_{KM}$ and $\betapar$ are less efficient than $\hat{\beta}_{PL}$ as expected.

\begin{table}[!h]
\centering
\caption{Comparison of $\hat{\beta}_{PL}$, $\hat{\beta}_{KM}$ and $\betapar$ under non-proportional hazards model. $C \sim \mathcal{U}[0,t_c]$. Standard errors in parenthesis.}
\label{tab2}
\begin{tabular}{|c|c|c|c|c|c|c|c|}
\hline
$\beta_1$ & $\beta_2$ & $t_0$ & $\%$ of censoring & $\hat{\beta}_{PL}$ & $\hat{\beta}_{KM}$ & $\betapar$ & $\mathbb{E}[\beta_0(T)]$\\
\hline
1 & 0 & 0.2 & 0\% & 0.330 (0.096) & 0.331 (0.096) & 0.330 (0.096) & 0.330 \\
 &  &  & 17\% & 0.373 (0.092) & 0.330 (0.101) & 0.329 (0.103) & 0.330 \\
 &  &  & 32\% & 0.418 (0.094) & 0.351 (0.122) & 0.348 (0.126) & 0.330 \\
 &  &  & 50\% & 0.512 (0.094) & 0.438 (0.159) & 0.437 (0.164) & 0.330 \\
\hline
3 & 0 & 0.2 & 0\% & 0.981 (0.110) & 0.984 (0.110) & 0.981 (0.112) & 0.989\\
 &  &  & 17\% & 1.123 (0.111) & 1.003 (0.119) & 0.999 (0.122) & 0.989\\
 &  &  & 32\% & 1.268 (0.116) & 1.073 (0.142) & 1.067 (0.149) & 0.989\\
 &  &  & 50\% & 1.569 (0.123) & 1.342 (0.196) & 1.337 (0.204) &  0.989\\
\hline 
\end{tabular}
\end{table}

\begin{table}[!h]
\centering
\caption{Comparison of $\hat{\beta}_{PL}$, $\hat{\beta}_{KM}$ and $\betapar$ under non-proportional hazards model. $C \sim \mathcal{E}(t_c)$. Standard errors in parenthesis.}
\label{tab3}
\begin{tabular}{|c|c|c|c|c|c|c|c|}
\hline
$\beta_1$ & $\beta_2$ & $t_0$ & $\%$ of censoring & $\hat{\beta}_{PL}$ & $\hat{\beta}_{KM}$ & $\betapar$ & $\mathbb{E}[\beta_0(T)]$\\
\hline
1 & 0 & 0.2 & 0\% & 0.331 (0.090) & 0.332 (0.090) & 0.331 (0.091) & 0.330 \\
 &  &  & 17\% & 0.382 (0.090) & 0.332 (0.104) & 0.331 (0.105) & 0.330 \\
 &  &  & 32\% & 0.444 (0.090) & 0.340 (0.112) & 0.337 (0.113) & 0.330 \\
 &  &  & 50\% & 0.549 (0.091) & 0.342 (0.161) & 0.338 (0.164) & 0.330 \\
\hline
3 & 0 & 0.2 & 0\% & 0.983 (0.113) & 0.986 (0.113) & 0.984 (0.116) & 0.989\\
 &  &  & 17\% & 1.135 (0.114) & 0.993 (0.124) & 0.989 (0.127) & 0.989\\
 &  &  & 32\% & 1.369 (0.121) & 1.061 (0.145) & 1.054 (0.148) & 0.989\\
 &  &  & 50\% & 1.769 (0.128) & 1.197 (0.197) & 1.186 (0.207) &  0.989\\
\hline 
\end{tabular}
\end{table}

Results for non-proportional hazards models are given in Table \ref{tab2} and Table \ref{tab3}. We find that $\hat{\beta}_{PL}$ depends on the censoring under the non-proportional hazards model because its value strongly varies when we increase the censoring rate. We can also notice that the estimators $\hat{\beta}_{KM}$ and $\betapar$ are consistent, even under non-proportional hazards model, whatever the percentage of censoring is. This was expected in light of Property \ref{Xu96} and Theorem \ref{NPH_bien_spe}.

Figure \ref{fig1} presents a comparison between the parametric estimator of the survival function presented in the previous section and the Kaplan-Meier estimator. The true survival function $S$ of $T$ is also drawn for more clarity. The chosen parameter $\beta_0$ for these simulations is $\beta_0= I(t<0.1)$, for all $t \in [0, \tau]$. The sample size is $n=100$. The distribution of $T$ is an exponential of parameter 2 for the first figure, and a Weibull of parameters 2 and 3 for the second one. The third distribution, denoted by $E_3$ is a piecewise exponential with parameters 0.25 on $[0,1[$, 1 on $[1,2[$ and 0.25 on $[2, + \infty[$. The final distribution, denoted by $E_4$ is also a piecewise exponential with parameters 0.25 on $[0,1[$, 1 on $[1,2[$, 2 on $[2,3[$ and 0.25 on $[3,+ \infty[$. In all cases, the percentage of censoring is 30\%. The parameters of these distributions are estimated by maximization of the likelihood.

\begin{figure}[!h]
\centering
\includegraphics[scale=0.6]{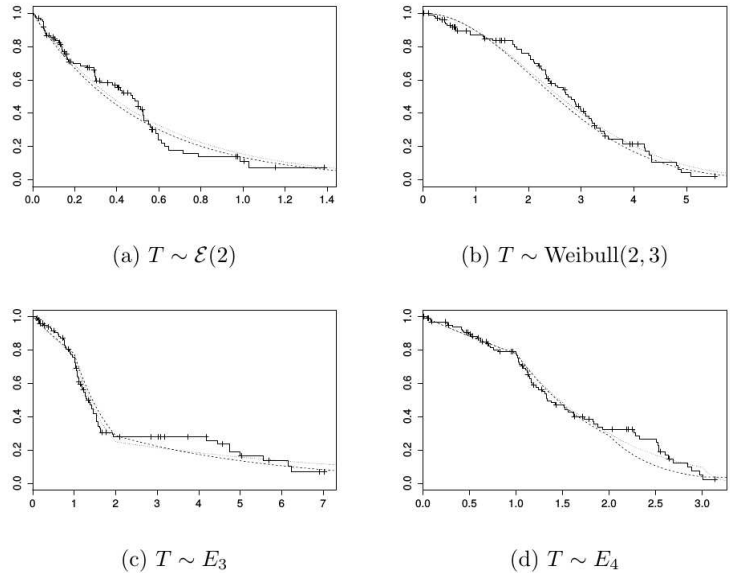}
\caption{Comparison between the Kaplan-Meier estimator (solid line), the parametric estimator (dotted line) of the survival function and the true survival curve (dashed line).}
\label{fig1}
\end{figure}

We can see that the parametric estimator is smoother than the Kaplan-Meier estimator, smoothing the increasingly large jumps for the large observed survival times.

We carried out some other simulations to compare the precisions of the parametric estimation and the Kaplan-Meier estimator under different $\beta_0$, different percentages of censoring and different sample sizes. It seems that the estimators $\hat{\beta}_{KM}$ and $\betapar$ have a similar precision for a given model and sample size. Therefore it can be interesting to use one estimator or another depending on the study. For example, we can give priority to $\betapar$ in case of supervised data.

\section{Relative efficiency under proportional hazards}\label{Rel-eff}

In Section \ref{LS_prop}, we established that, under proportional hazards models, $\betapar$ is consistent for $\beta_0$, which is a constant ; so is $\hat{\beta}_{PL}$. Then a natural question is to find which one of these two estimators is the most efficient. We already know that $\hat{\beta}_{PL}$ is the most efficient estimator under the Cox's model \citep{efron1977}. The question that remains to be addressed is how less efficient $\betapar$ is.

The results of \citet{lin1991} lead to the asymptotic relative efficiency of $\betapar$ to $\hat{\beta}_{PL}$:
\[
R_{eff}(\betapar,\hat{\beta}_{PL}) = \frac{\left( \Sigma_1 \right)^2}{\Sigma_0 \Sigma_2},
\]
where
\[
\Sigma_0= \int_0^\infty v(\beta_0,t) s^{(0)}(\beta_0,t)dt , \ \Sigma_1 = \int_0^\infty v(\beta_0,t) dF(t) , \ \Sigma_2 = \int_0^\infty v(\beta_0,t) \frac{S(t)}{s^{(0)}(0,t)}dF(t).
\]
To illustrate this, we study the following case: the baseline hazard is constant equal to 1, $Z$ has a Bernoulli distribution with parameter $p$ and $C$ has a lognormal distribution with parameters 0 and $t_c$. The latter distribution is chosen to ensure the convergence of the integrals $\Sigma_i$. We have then
\[
\Sigma_0 = \int_0^{\infty} A(\beta_0,t) \text{P}(C \geq t)dt,\ \Sigma_1 = \int_0^{\infty} A(\beta_0,t)dt, \ \Sigma_2 = \int_0^{\infty} \frac{A(\beta_0,t)}{\text{P}(C \geq t)}dt,
\]
where
\[
A(\beta,t)=\frac{(1-p)e^{-t} pe^\beta \exp(-te^\beta)}{(1-p)e^{-t}+ pe^\beta \exp(-te^\beta)}.
\]
Results for several values of $p$, $\beta_0$ and $t_c$ are presented in Table \ref{tab4}.

\begin{table}[!h]
\centering
\caption{Asymptotic relative efficiency of $\betapar$ to $\hat{\beta}_{PL}$ under proportional hazards. Percentage of censoring in parenthesis.}
\begin{tabular}{|c|c|c|c|c|}
\hline
$\beta_0$ & $t_c$ & $p=0.25$ & $p=0.5$ & $p=0.75$\\
\hline
0.5 & 1 & 0.797 (35\%) & 0.772 (32\%) & 0.736 (29\%) \\
\hline
1 &  & 0.911 (33\%) & 0.892 (27\%) &  0.863 (21\%) \\
\hline
2 &  & 0.990 (30\%) & 0.986 (21\%) &  0.979 (12\%) \\
\hline
0.5 & 0.5 & 0.192 (35\%) & 0.150 (32\%) & 0.100 (29\%) \\
\hline
1 &  & 0.746 (33\%) & 0.675 (27\%) &  0.564 (21\%) \\
\hline
2 &  & 0.996 (30\%) & 0.993 (21\%) &  0.988 (12\%) \\
\hline 
\end{tabular}
\label{tab4}
\end{table}

We can notice that the asymptotic relative efficiency of $\betapar$ to $\hat{\beta}_{PL}$ can be poor for a heavy censoring mechanism (fourth row of Table \ref{tab4}). Moreover, the larger the value of $\beta_0$, the larger the asymptotic relative efficiency. To convince ourselves of this, note that, when $\vert \beta_0 \vert \to \infty$, $A(\beta_0,t) \to 0$ and so $R_{eff}(\betapar,\hat{\beta}_{PL}) \to 1$.

\section{An illustration of the variance of $\betapar$ on the Freireich dataset}

A natural question at this stage is to determine the variance of $\betapar$ and compare it to the variance of $\hat{\beta}_{PL}$. In order to visualize this comparison on a real case, we use the leukemia data of \citet{freireich1963}. We obtain $\hat{\beta}_{PL} \approx 1.56$ and $\betapar \approx 1.59$. We generated 1000 resampling estimators by introducing random weights in the estimating equations. To be more precise, we generated 500 $n_t$-samples with an exponential distribution of parameter 1, where $n_t$ is the number of failure times in the data. We denote one of them by $(e_1, \cdots, e_{n_t})$. Then we put weights of the form $e_i/\sum_{j=1}^{n_t} e_j$ in the estimating equations (\ref{eq_estim_PL}) and (\ref{eq_estimW}). We obtained two histograms presented in Figure \ref{fig2}. We also draw the theoretical asymptotic distributions of $\hat{\beta}_{PL}$ and $\betapar$ : a gaussian distribution with mean $\hat{\beta}_{PL}$ and variance given by \citet{andersen-gill1982} for $\hat{\beta}_{PL}$, and a gaussian distribution with mean $\betapar$ and variance given by \citet{lin1991} for $\betapar$.  This latter variance can be consistently estimated by $A (\betapar )^{-1} B ( \betapar ) A ( \betapar )^{-1}$, where 
\[
A(\beta)= \frac{1}{n} \sum_{i=1}^n \Delta_i W_{\text{par}}(X_i) V(\beta, X_i) \mbox{ and } 
B(\beta)= \frac{1}{n} \sum_{i=1}^n \Delta_i W_{\text{par}}^2(X_i) V(\beta, X_i).
\]

\begin{figure}[!h]
\centering
\includegraphics[scale=0.6]{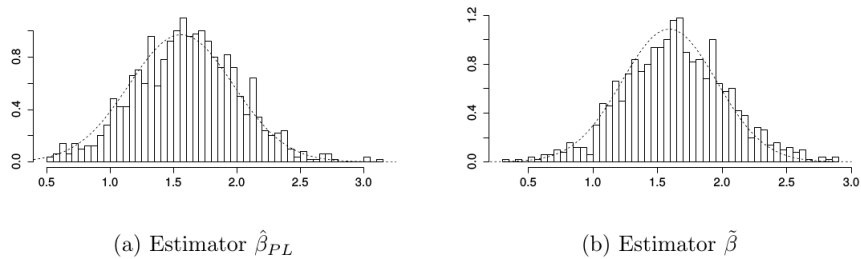}
\caption{Comparison between the histograms of 1000 resampling estimators and the theoretical gaussian distribution (dotted line).}
\label{fig2}
\end{figure}

The estimated standard error of $\hat{\beta}_{PL}$ is 0.42 and the estimated standard error of $\betapar$ is 0.36. We can see, in this configuration, that $\betapar$ and $\hat{\beta}_{PL}$ have very close empirical and theoretical distributions. These data are generally taken to satisfy the Cox proportional hazards model and it is known that $\hat{\beta}_{PL}$ is the best estimator of the true regression parameter $\beta_0$ under such a model. This result shows the value of using $\betapar$ under a proportional hazards model and the validity of its use under a non-proportional hazards model has been shown in Section \ref{LS_prop}.







\section{An application to relative survival data}\label{Appli}

In this section, we apply the method of Section \ref{Inference} on acute myocardial infarction data collected at the University Clinical Center in Ljubljana, Slovenia. The data are included in the R package \texttt{relsurv} in the table \texttt{rdata} developped by \citet{pohar-stare2006}. We use the Slovenian population tables contained in the table \texttt{slopop} to estimate the distribution of $T$. The empirical distribution of $T$ is presented in Figure \ref{fig3}.  We choose a piecewise exponential with one changepoint at 1461 days. 

\begin{figure}[!h]
\centering
\includegraphics[scale=0.3]{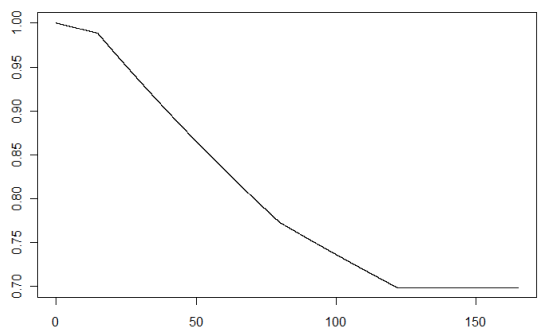}
\caption{Empirical distribution of $T$ obtained by Slovenian population tables (solid line). Parametric estimator of $S(t)$ with a piecewise exponential (dashed line). Time is expressed in days from 1rst January 1982.}
\label{fig3}
\end{figure}

Then we computed $\betapar$ and $\hat{\beta}_{PL}$ for the following covariates : age, sex, year of diagnosis and age category ("under 54", "54-61", "62-70", "71-95"). We generated 500 bootstrap \citep{efron-tibshirani1994} samples based on the data of the University Clinical Center in Ljubljana in order to get estimations for the standard errors. The results are summarized in Table \ref{tab5}.

\begin{table}[!h]
\centering
\caption{Average effects estimated from acute myocardial infarction data. Standard errors in parenthesis.}
\label{tab5}
\begin{tabular}{|c|c|c|c|c|c|}
\hline
 & Age & Sex & Year & Age category \\
\hline
$\betapar$ & 0.046 (0.012) & -0.069 (0.358) & 0.005 (0.001) & 0.444 (0.144)  \\
$\hat{\beta}_{PL}$ & 0.061 (0.004) & 0.530 (0.089) & 0.000 (0.0001) & 0.573 (0.044)\\
\hline
\end{tabular}
\end{table}

Alternatively, we could assume that $\beta_0$ is piecewise constant and to estimate it on each piece. For this, we can use changepoint methods, although this is not the focus of this article. Details are given in \citet{xu-adak2002}. What we get here are the average effects of the different studied covariables.

\section{Discussion}

In this paper, we study the effects of knowing the marginal distribution of $T$ on the estimation of $\beta^ *$, a summary statistic for $\beta_0(.)$. This extra knowledge does not help us increase efficiency in the case of proportional hazards. On the other hand, in the presence on non-proportional hazards, the estimator converges in probability to a population counterpart that has strong interpretation as average effect. \citet{xu-oquigley2000} obtained a similar result, under certain conditions, using the Kaplan-Meier estimator. It is useful to know that the result of \citet{xu-oquigley2000} extends to the estimation where the marginal survival can be modeled parametrically. Note that this is much weaker than assuming parametric models for the whole family of conditional distributions given the covariates. Indeed, even when the model for marginal survival is misspecified, we still obtain consistent estimators of the regression coefficient, as long as the proportional hazards assumption continues to hold. We have not studied the impact of departures from the proportional hazards assumption in conjunction with a misspecified marginal model for survival. This may be worthy of further study.


\section*{Appendix}

We state the following Lemma which is a straightforward probabilist version of the implicit function theorem that we can find in \citet{schwartz1992} for instance.

\begin{lemma}\label{TFI_prob}{Probabilist implicit function theorem}\\
Let $\U$ be an open set of $\R^m \times \R^p$, $k \in \N^*$ and $X$ an almost surely $\mathcal{C}^k$-process defined on $\U$ with values in $\R^p$. Let $(a,b) \in \R^m \times \R^p$ such that $X(a,b)=0$ and $D_tX(a,b)$ is invertible almost surely.

Then there exists a neighbourhood $\V$ of $a$ in $\R^m$, a neighbourhood $\W$ of $b$ in $\R^p$ and a process $\phi : \V \longrightarrow \W$ almost surely $\mathcal{C}^k$ such that $\V \times \W \subset \U$ and
\begin{equation*}
\forall s \in \V, \ \forall t \in \W, \ \text{almost surely}, \ X(s,t)=0 \Leftrightarrow t= \phi(s)
\end{equation*}
Furthermore, almost surely,
\begin{equation}\label{diff-fct-imp}
\forall s \in \V, \ d\phi(s)=-D_tX(s,\phi(s))^{-1} \circ D_sX(s,\phi(s)).
\end{equation}
\end{lemma}

\bibliographystyle{natbib} 
\bibliography{biblio}

\begin{thebibliography}{27}
\providecommand{\natexlab}[1]{#1}
\providecommand{\url}[1]{\texttt{#1}}
\expandafter\ifx\csname urlstyle\endcsname\relax
  \providecommand{\doi}[1]{doi: #1}\else
  \providecommand{\doi}{doi: \begingroup \urlstyle{rm}\Url}\fi

\bibitem[Andersen and Gill(1982)]{andersen-gill1982}
Andersen, P.~K. and Gill, R.~D.
\newblock (1982).
\newblock Cox's regression model for counting processes: A large sample study.
\newblock \emph{The Annals of Statistics}, 10\penalty0 (4):\penalty0 pp.
  1100--1120.

\bibitem[Cox(1972)]{cox1972}
Cox, D.~R.
\newblock (1972).
\newblock Regression models and life-tables.
\newblock \emph{Journal of the Royal Statistical Society. Series B
  (Methodological)}, 34\penalty0 (2):\penalty0 pp. 187--220.

\bibitem[Efron(1977)]{efron1977}
Efron, B.
\newblock (1977).
\newblock The efficiency of cox's likelihood function for censored data.
\newblock \emph{Journal of the American Statistical Association}, 72\penalty0
  (359):\penalty0 557--565.

\bibitem[Efron and Tibshirani(1994)]{efron-tibshirani1994}
Efron, B. and Tibshirani, R.J.
\newblock \emph{An Introduction to the Bootstrap}.
\newblock Chapman \& Hall/CRC Monographs on Statistics \& Applied Probability.
  Taylor \& Francis, (1994).
\newblock ISBN 9780412042317.

\bibitem[Freireich et~al.(1963)Freireich, Gehan, Frei, Schroeder, Wolman,
  Anbari, Burgert, Mills, Pinkel, Selawry, et~al.]{freireich1963}
Freireich, E.~J., E.~Gehan, E.~Frei, L.~R. Schroeder, I.~J. Wolman, R.~Anbari,
  E.~O. Burgert, S.~D. Mills, D.~Pinkel, O.~S. Selawry, and others, .
\newblock (1963).
\newblock The effect of 6-mercaptopurine on the duration of steroid-induced
  remissions in acute leukemia: A model for evaluation of other potentially
  useful therapy.
\newblock \emph{Blood}, 21\penalty0 (6):\penalty0 699--716.

\bibitem[Gehan(1965)]{gehan1965}
Gehan, E.~A.
\newblock (1965).
\newblock A generalized wilcoxon test for comparing arbitrarily singly-censored
  samples.
\newblock \emph{Biometrika}, 52\penalty0 (1-2):\penalty0 203--223.

\bibitem[Gray(1992)]{gray1992}
Gray, R.~J.
\newblock (1992).
\newblock Flexible methods for analyzing survival data using splines, with
  applications to breast cancer prognosis.
\newblock \emph{Journal of the American Statistical Association}, 87\penalty0
  (420):\penalty0 942--951.

\bibitem[Harrington and Fleming(1982)]{harrington-fleming1982}
Harrington, D.~P. and Fleming, T.~R.
\newblock (1982).
\newblock A class of rank test procedures for censored survival data.
\newblock \emph{Biometrika}, 69\penalty0 (3):\penalty0 553--566.

\bibitem[Hastie and Tibshirani(1993)]{hastie-tibshirani1993}
Hastie, T. and Tibshirani, R.
\newblock (1993).
\newblock Varying-coefficient models.
\newblock \emph{Journal of the Royal Statistical Society. Series B
  (Methodological)}, 55\penalty0 (4):\penalty0 757--796.

\bibitem[Kaplan and Meier(1958)]{kaplan-meier1958}
Kaplan, E.~L. and Meier, Paul.
\newblock (1958).
\newblock Nonparametric estimation from incomplete observations.
\newblock \emph{Journal of the American Statistical Association}, 53\penalty0
  (282):\penalty0 pp. 457--481.

\bibitem[Lausen and Schumacher(1996)]{lausen-schumacher1996}
Lausen, B. and Schumacher, M.
\newblock (1996).
\newblock Evaluating the effect of optimized cutoff values in the assessment of
  prognostic factors.
\newblock \emph{Comput. Stat. Data Anal.}, 21\penalty0 (3):\penalty0 307--326.

\bibitem[Liang et~al.(1990)Liang, Self, and Liu]{liang-self-liu1990}
Liang, K.~Y., S.~G. Self, and Liu, X.
\newblock (1990).
\newblock The cox proportional hazards model with change point: An
  epidemiologic application.
\newblock \emph{Biometrics}, 46\penalty0 (3):\penalty0 783--793.

\bibitem[Lin(1991)]{lin1991}
Lin, D.~Y.
\newblock (1991).
\newblock Goodness-of-fit analysis for the cox regression model based on a
  class of parameter estimators.
\newblock \emph{Journal of the American Statistical Association}, 86\penalty0
  (415):\penalty0 pp. 725--728.

\bibitem[Marzec and Marzec(1997)]{marzec-marzec1997}
Marzec, L. and Marzec, P.
\newblock (1997).
\newblock On fitting cox's regression model with time-dependent coefficients.
\newblock \emph{Biometrika}, 84\penalty0 (4):\penalty0 901--908.

\bibitem[Moreau et~al.(1985)Moreau, O'Quigley, and
  Mesbah]{moreau-oquigley-mesbah1985}
Moreau, T., J.~O'Quigley, and Mesbah, M.
\newblock (1985).
\newblock A global goodness-of-fit statistic for the proportional hazards
  model.
\newblock \emph{Journal of the Royal Statistical Society. Series C (Applied
  Statistics)}, 34\penalty0 (3):\penalty0 212--218.

\bibitem[Murphy and Sen(1991)]{murphy-sen1991}
Murphy, S.~A. and Sen, P.~K.
\newblock (1991).
\newblock Time-dependent coefficients in a cox-type regression model.
\newblock \emph{Stochastic Processes and their Applications}, 39\penalty0
  (1):\penalty0 153 -- 180.

\bibitem[O'Quigley and Pessione(1989)]{oquigley-pessione1989}
O'Quigley, J. and Pessione, F.
\newblock (1989).
\newblock Score tests for homogeneity of regression effect in the proportional
  hazards model.
\newblock \emph{Biometrics}, 45\penalty0 (1):\penalty0 135--144.

\bibitem[O'Quigley and Pessione(1991)]{oquigley-pessione1991}
O'Quigley, J. and Pessione, F.
\newblock (1991).
\newblock The problem of a covariate-time qualitative interaction in a survival
  study.
\newblock \emph{Biometrics}, 47\penalty0 (1):\penalty0 101—115.

\bibitem[Pohar and Stare(2006)]{pohar-stare2006}
Pohar, M. and Stare, J.
\newblock (2006).
\newblock Relative survival analysis in r.
\newblock \emph{Computer methods and programs in biomedicine}, 81\penalty0
  (3):\penalty0 272--278.

\bibitem[Schwartz(1992)]{schwartz1992}
Schwartz, L.
\newblock \emph{Analyse. 2. Calcul diff{\'e}rentiel et {\'e}quations
  diff{\'e}rentielles}.
\newblock Hermann, (1992).

\bibitem[Struthers and Kalbfleisch(1986)]{struthers-kalbfleisch1986}
Struthers, C.~A. and Kalbfleisch, J.~D.
\newblock (1986).
\newblock Misspecified proportional hazard models.
\newblock \emph{Biometrika}, 73\penalty0 (2):\penalty0 pp. 363--369.

\bibitem[Stute and Wang(1993)]{stute-wang1993}
Stute, W. and Wang, J.-L.
\newblock (1993).
\newblock The strong law under random censorship.
\newblock \emph{Ann. Statist.}, 21\penalty0 (3):\penalty0 1591--1607.

\bibitem[Verweij and Houwelingen(1995)]{verweij-houwelingen1995}
Verweij, P. and Houwelingen, H.~Van.
\newblock (1995).
\newblock Time-dependent effects of fixed covariates in cox regression.
\newblock \emph{Biometrics}, 51\penalty0 (4):\penalty0 1550--1556.

\bibitem[Xu and Adak(2002)]{xu-adak2002}
Xu, R. and Adak, S.
\newblock (2002).
\newblock Survival analysis with time-varying regression effects using a
  tree-based approach.
\newblock \emph{Biometrics}, 58\penalty0 (2):\penalty0 305--315.
\newblock ISSN 1541-0420.

\bibitem[Xu and Harrington(2001)]{xu-harrington2001}
Xu, R. and Harrington, D.~P.
\newblock (2001).
\newblock A semiparametric estimate of treatment effects with censored data.
\newblock \emph{Biometrics}, 57\penalty0 (3):\penalty0 875--885.

\bibitem[Xu and O'Quigley(2000)]{xu-oquigley2000}
Xu, R. and O'Quigley, J.
\newblock (2000).
\newblock Estimating average regression effect under non-proportional hazards.
\newblock \emph{Biostatistics}, 1\penalty0 (4):\penalty0 423--439.

\bibitem[Zucker and Karr(1990)]{zucker-karr1990}
Zucker, D.~M. and Karr, A.~F.
\newblock (1990).
\newblock Nonparametric survival analysis with time-dependent covariate
  effects: a penalized partial likelihood approach.
\newblock \emph{Ann. Statist.}, 18\penalty0 (1):\penalty0 329--353.

\end{thebibliography}



\end{document}